\documentclass[conference,letterpaper]{IEEEtran}

\IEEEoverridecommandlockouts

\usepackage{cite}
\usepackage{url}
\usepackage{ifthen}
\usepackage{algorithm, algorithmicx, algpseudocode}
\usepackage[dvipdfmx]{graphicx}
\usepackage{textcomp}
\usepackage[cmex10]{amsmath}
\usepackage{amssymb,amsfonts}
\usepackage{pifont}
\usepackage{amsthm}
\usepackage{mathrsfs}
\def\BibTeX{{\rm B\kern-.05em{\sc i\kern-.025em b}\kern-.08em T\kern-.1667em\lower.7ex\hbox{E}\kern-.125emX}}
\usepackage{tikz}
\usetikzlibrary{automata,positioning,calc}
\usetikzlibrary{intersections}
\usetikzlibrary{decorations.pathreplacing,angles,quotes}

\usepackage{bm}
\usepackage{amscd}

\theoremstyle{definition}
\newtheorem{theorem}{Theorem}
\newtheorem{definition}{Definition}

\newtheorem{prop}{Proposition}
\newtheorem{cor}{Corollary}
\newtheorem{remark}{Remark}
\newtheorem{eg}{Example}

\algnewcommand{\Initialize}[1]{%
  \State \textbf{Initialization:}
  \Statex \hspace*{\algorithmicindent}\parbox[t]{0.8\linewidth}{\raggedright #1}
}

\newcommand{\abs}[1]{\left\lvert#1\right\rvert}

\newcommand{\one}[1]{\mbox{1}\hspace{-0.25em}\mbox{l}_{\left\{#1\right\}}}

\newcommand{\argmin}{\operatornamewithlimits{argmin}}
\newcommand{\arginf}{\operatornamewithlimits{arginf}}

\newcommand{\relmiddle}[1]{\mathrel{}\middle#1\mathrel{}} 
\newcommand\independent{\protect\mathpalette{\protect\independenT}{\perp}}
\def\independenT#1#2{\mathrel{\rlap{$#1#2$}\mkern2mu{#1#2}}}

\def\br{\mathbb R}

\def\vE{\mathbb E}

\font\b=cmr10 scaled\magstep4

\def\bigzerou{\smash{\lower1.7ex\hbox{\b 0}}}
\def\bigzerou{\smash{\lower1.7ex\hbox{\b 0}}}

\begin{document}

\title{An Algorithm for Computing the Stratonovich's Value of Information
\thanks{This research is supported in part by Grant-in-Aid JP17K06446 for Scientific Research (C).}
}

\author{
\IEEEauthorblockN{Akira Kamatsuka}
\IEEEauthorblockA{Shonan Institute of Technology \\ 
Email: \text{kamatsuka@info.shonan-it.ac.jp}
 }
\and 
\IEEEauthorblockN{Takahiro Yoshida}
\IEEEauthorblockA{Nihon University \\ 
Email: \text{yoshida.takahiro@nihon-u.ac.jp}
  }
\and
\IEEEauthorblockN{Koki Kazama, Toshiyasu Matsushima}
\IEEEauthorblockA{Waseda University \\ 
Email: \text{kokikazama@aoni.waseda.jp}, \\ \text{toshimat@waseda.jp}
 }
}

\maketitle

\begin{abstract}
We propose an algorithm for computing Stratonovich's value of information (VoI) that can be regarded as an analogue of the distortion-rate function. 
We construct an alternating optimization algorithm for VoI under a general information leakage constraint and derive a convergence condition. 
Furthermore, we discuss algorithms for computing VoI under specific information leakage constraints, 
such as Shannon's mutual information (MI), $f$-leakage, Arimoto's MI, Sibson's MI, and Csisz{\'a}r's MI.
\end{abstract}

\section{Introduction}

Decision-making based on noisy data has recently been studied extensively in the field of machine learning and information theory.
Examples include the information disclosure problem with privacy protection (e.g. \cite{8804205}) and 
the classification problem in the presence of label noise (e.g. \cite{6685834}).

Such research can be traced back to the theory of \textit{value of information} (VoI) 
pioneered by Stratonovich in the 1960s \cite{Stratonovich1965}\footnote{Recently, his book containing this theory has been translated into English \cite{belavkin2020theory}.}. 
He analyzed the inference gain when using the noisy data $Y$ containing 
at most $R$ [bits] of mutual information $I(X;Y)$ about the original data $X$ 
and derived a theoretical result on the fundamental trade-off between the amount of mutual information and the inferential gain.
Recently, this result has been extended by us to a general information leakage measure that is not limited to mutual information \cite{https://doi.org/10.48550/arxiv.2201.11449}. 
In \cite{https://doi.org/10.48550/arxiv.2201.11449}, We also gave an interpretation of the result in terms of optimal privacy mechanism in the privacy-utility trade-off (PUT) problem.

Rate-distortion theory developed by Shannon \cite{Shannon:1959:CTD},  on the other hand, is a well-known trade-off problem in information theory, 
and various theoretical studies have been conducted.
In particular, an alternating optimization algorithm was proposed  by Blahut as a computational algorithm for the rate-distortion function \cite{1054855}, 
which is now called the Arimoto-Blahut algorithm.
Later, Csisz\'ar and Tusn{\'{a}}dy analyzed the convergence properties of the alternating optimization algorithm in terms of information geometry \cite{10008948102},\cite{Csiszar:2004:ITS:1166379.1166380}. 
Furthermore, a simpler sufficient convergence condition was derived by Yueng \cite{531176},\cite{10.5555/1199866}.

Inspired by these results in the rate-distortion theory, we consider applying the alternating optimization algorithm 
for computing VoI under the general information leakage measure constraint.
This algorithm allows us to construct an optimal privacy mechanism in the PUT problem.

Our main contributions are as follows:

\begin{itemize}
\item We provide an alternating optimization algorithm framework for computing VoI under a general information leakage constraint (Algorithm \ref{alg:ab_like}) and derive a convergence condition to a globally optimal solution (Theorem \ref{thm:fundamental}, Corollary \ref{cor:convergence}).
\item We consider alternating optimization algorithms for VoI under Shannon's MI, $f$-leakage, Arimoto's MI, Sibson's MI, and Csisz{\'a}r's MI constraints. 
Then we derive and discuss the KKT conditions for them (Section \ref{sec:application}).
\end{itemize}

\section{Preliminary}\label{sec:preliminary}
We first review the theory of the \text{Value of Information} (VoI) \cite{Stratonovich1965}, \cite{https://doi.org/10.48550/arxiv.2201.11449} on the system model in Figure \ref{fig:system_model} and the alternating optimization problem \cite{531176},\cite{10.5555/1199866}.

\subsection{Notations}

Let $X, Y$ and $A$ be discrete random variables on finite alphabets 
$\mathcal{X}, \mathcal{Y}$ and $\mathcal{A}$.
$X$ and $Y$ represent the original data and the noisy data respectively, 
while $A$ represents an action. 
Let $p_{X, Y} = p_{X}\times p_{Y\mid X}$ be a given joint distribution of $(X, Y)$. 
Let $\delta\colon \mathcal{Y} \to \mathcal{A}$ be a deterministic decision rule 
and $\ell(x, a)(\geq 0)$ be a non-negative loss function which represents 
a loss for making an action $A=a$ when the true state is $X=x$. 
We use $\vE_{X}[f(X)]$ and $\vE_{X}[f(X) | Y=y]$ to represent expectation on $f(X)$ and conditional expectation on $f(X)$ given $Y=y$, respectively, where $f(X)$ is a function of $X$. 
We also use $\vE_{X}^{p_{X}}[f(X)]$ to emphasize that we are taking expectations in $p_{X}$.
Finally, we use $\log$ to represent the natural logarithm.

\begin{figure}[htbp]
\centering
\includegraphics[width=2.4in, clip]{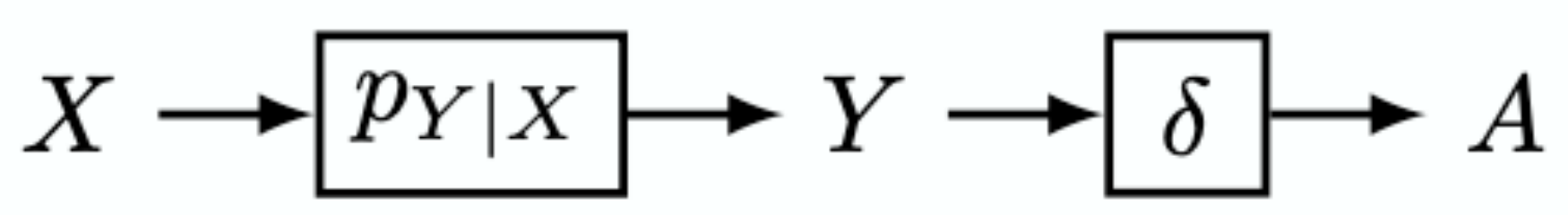}
\caption{System model}
\label{fig:system_model}
\end{figure}

\subsection{Stratonovich's Value of Information (VoI)}
In \cite{https://doi.org/10.48550/arxiv.2201.11449}, 
we introduced a general information leakage measure in an axiomatic way as follows.
\begin{definition} [\text{\cite[Def. 3]{https://doi.org/10.48550/arxiv.2201.11449}}]\label{def:info_leakage}
The information leakage $\mathcal{L}(X\to Y) = \mathcal{L}(p_{X}, p_{Y\mid X})$ is defined as a 
functional of $p_{X}$ and $p_{Y\mid X}$ that satisfies following properties:
\begin{enumerate}
\item \textit{Non-negativity}: 
\begin{align}
\mathcal{L}(X\to Y) \geq 0.
\end{align}
\item \textit{Data Processing Inequality (DPI)}: 

If $X-Y-Z$ forms a Markov chain, then 
\begin{align}
\mathcal{L}(X\to Z) \leq \mathcal{L}(X\to Y).
\end{align}
\item \textit{Independence}:
\begin{align}
\mathcal{L}(X\to Y) = 0 \Longleftrightarrow X \independent Y.  
\end{align}
\end{enumerate}
\end{definition}
We also assume that the information leakage $\mathcal{L}(X\to Y)$ is bounded above, i.e., 
there exists an upper bound $K(X)$ that can depend on $p_{X}$ such that for all $p_{Y\mid X}$, $\mathcal{L}(X\to Y) \leq K(X)$.

\begin{eg} Some examples of the information leakage are Shannon's mutual information (MI) \cite{shannon} $I(X; Y) := H(X) - H(X|Y)$ ,  
$f$-leakage \cite[Def. 7]{8804205} $\mathcal{L}_{f}(X\to Y) := \min_{q_{Y}>0} D_{f}(p_{X, Y} || p_{X}\times q_{Y})$, 
Arimoto's MI of order $\alpha$ \cite{arimoto1977} $I_{\alpha}^{\text{A}}(X; Y) := H_{\alpha}(X) - H_{\alpha}^{\text{A}}(X| Y)$,  
Sibson's MI of order $\alpha$ \cite{Sibson1969InformationR} $I_{\alpha}^{\text{S}}(X; Y) := \min_{q_{Y}>0} D_{\alpha}(p_{X, Y} || p_{X}\times q_{Y})$, 
and Csisz\'ar's MI of order $\alpha$ \cite{370121} $I_{\alpha}^{\text{C}}(X; Y) := \min_{q_{Y}>0} \vE_{X}\left[D_{\alpha}(p_{Y\mid X}(\cdot\mid X) || q_{Y})\right]$, 
where $\alpha \in (0, 1) \cup (1, \infty)$, the minimums are over all distributions $q_{Y}$ that satisfy $q_{Y}(y)>0$ for all $y\in \mathcal{Y}$, 
$H(X):=-\sum_{x}p_{X}(x)\log p_{X}(x)$ is the Shannon entropy of $X$, $H(X| Y):=-\sum_{x,y}p_{X}(x)p_{X\mid Y}(y|x)\log p_{X\mid Y}(x|y)$ is the 
conditional entropy of $X$ given $Y$, 
$H_{\alpha}(X) := \frac{\alpha}{1-\alpha}\log \left(\sum_{x} p_{X}(x)^{\alpha} \right)^{\frac{1}{\alpha}}$ is the R\'enyi entropy of $X$ of order $\alpha$, 
$H_{\alpha}^{\text{A}}(X | Y) :=  \frac{\alpha}{1-\alpha} \log \sum_{y} \left( \sum_{x} p_{X, Y}(x, y)^{\alpha} \right)^{\frac{1}{\alpha}}$   
is Arimoto's conditional entropy of $X$ given $Y$ of order $\alpha$, 
and $D_{f}(p || q) := \sum_{z\in \mathcal{\mathcal{Z}}} q(z)f \left( \frac{p(z)}{q(z)} \right)$ is the $f$-divergence, where 
$f\colon[0, \infty)\to \br$ is a convex function such that $f(1) = 0$, strictly convex at $t=1$. 
Table \ref{tab:f-leakage} shows a list of $f$-divergence.
\end{eg}

\begin{table*}[h]
   \centering
   \resizebox{.8\textwidth}{!}{
  \begin{tabular}{@{} |c||c|c|c|c| @{}}
    \hline
    \multicolumn{1}{|c||}{$f(t)$} 
    & \multicolumn{1}{|c|}{Name of $D_{f}(p || q)$} 
    & \multicolumn{1}{c|}{$q_{A}^{*}= \argmin_{q_{A}} D_{f}(p_{X}p_{A\mid X} || p_{X}q_{A})$}
    \\ \hline \hline 
    $t\log t$
    & KL-divergence 
    & $\sum_{x}p_{X}(x)p_{A\mid X}(a\mid x)$ 
    \\ \hline 
    $-\log t$
    & reverse KL-divergence 
    & $\frac{\exp\{\sum_{x}p_{X}(x)\log p_{A\mid X}(a\mid x)\}}{\sum_{a}\exp\{\sum_{x}p_{X}(x)\log p_{A\mid X}(a\mid x)\}}$ 
    \\ \hline 
    $2(\sqrt{t}-1)$
    & squared Hellinger distance 
    & $\frac{\sum_{x}\sqrt{p_{X}(x)^{2}p_{A\mid X}(a\mid x)}}{\sum_{x,a}\sqrt{p_{X}(x)^{2}p_{A\mid X}(a\mid x)}}$
    \\ \hline 
    $(t-1)^{2}$
    & Pearson $\chi^{2}$-divergence
    & $\frac{\sqrt{\sum_{x}p_{X}(x)p_{A\mid X}(a\mid x)^{2}}}{\sum_{a}\sqrt{\sum_{x}p_{X}(x)p_{A\mid X}(a\mid x)^{2}}}$ 
    \\ \hline 
    $1/t-1$
    & Neyman $\chi^{2}$-divergence
    & $\frac{\left( \sum_{x}p_{X}(x)p_{A\mid X}(a\mid x)^{-1}\right)^{-1}}{\sum_{a}\left( \sum_{x}p_{X}(x)p_{A\mid X}(a\mid x)^{-1}\right)^{-1}}$ 
    \\ \hline 
    $(t^{\alpha}-1)/(\alpha-1)$
    & Hellinger divergence of order $\alpha$
    & $\frac{\left( \sum_{x}p_{X}(x)p_{A\mid X}(a\mid x)^{\alpha} \right)^{1/\alpha}}{\sum_{a}\left( \sum_{x}p_{X}(x)p_{A\mid X}(a\mid x)^{\alpha} \right)^{1/\alpha}}$
    \\ \hline 
    $4(1-t^{(\alpha+1)/2})/(1-\alpha^{2})$
    & $\alpha$-divergence
    & $\frac{\left( \sum_{x}p_{X}(x)p_{A\mid X}(a\mid x)^{(\alpha+1)/2} \right)^{2/(\alpha+1)}}{\sum_{a}\left( \sum_{x}p_{X}(x)p_{A\mid X}(a\mid x)^{(\alpha+1)/2} \right)^{2/(\alpha+1)}}$ 
    \\ \hline 
  \end{tabular}
  }
  \caption{List of $f$-divergence and its minimizer}
  \label{tab:f-leakage}
\end{table*}

Stratonovich introduced the following quantity, which we term as \textit{average gain}, 
to quantify the inferential gain of using the noisy data  $Y$ for a decision-making as largest reduction of the minimal expected loss compared to no-data situation.

\begin{definition}[Average gain \text{\cite[Def. 5]{https://doi.org/10.48550/arxiv.2201.11449}}] \label{def:ave_gain} 
The average gain of using $Y$ on $X$ for making an action $A$ with a loss function $\ell(x, a)$ 
is defined as
\begin{align}
\textsf{gain}^{\ell}(X; Y) &= \inf_{a} \vE_{X}\left[\ell(X, a)\right]  - \vE_{Y}\left[\inf_{a} \vE_{X}\left[\ell(X, a) \mid Y\right] \right]. 
\end{align}
\end{definition}

Then VoI is defined as follows.

\begin{definition}[VoI \text{\cite[Def. 7]{https://doi.org/10.48550/arxiv.2201.11449}}] \label{def:VoI_for_information_leakage} 
VoI for a loss function $\ell(x, a)$ and a information leakage measure $\mathcal{L}(X\to Y)$ is given as
\begin{align}
&\textsf{V}_{\mathcal{L}}^{\ell}(R; \mathcal{Y}) 
:= \sup_{\substack{p_{Y\mid X}\colon \\ 
{\mathcal{L}(X\to Y)}\leq R}} \textsf{gain}^{\ell}(X; Y) \\
&=\inf_{a} \vE_{X}\left[\ell(X, a)\right] 
- \inf_{\substack{p_{Y\mid X}\colon \\ 
{\mathcal{L}(X\to Y)}\leq R}}\vE_{Y}\left[\inf_{a} \vE_{X}\left[\ell(X, a) \mid Y\right] \right]. 
\end{align}
\end{definition}

Stratonovich first proved the fundamental trade-off between the amount of information leakage and inferential gain showing the 
following achievable upper bound $\textsf{V}_{\mathcal{L}}^{\ell}(R)$ for VoI $\textsf{V}_{\mathcal{L}}^{\ell}(R; \mathcal{Y})$ \cite[Chapter 9.7]{belavkin2020theory}, 
which is extended by us \text{\cite[Thm. 1]{https://doi.org/10.48550/arxiv.2201.11449}}.
\begin{theorem}[\text{\cite[Thm. 1]{https://doi.org/10.48550/arxiv.2201.11449}}] \label{thm:main_result_for_information_leakage}
For a loss function $\ell(x,a)$, define a function as follows:
\begin{align} 
\textsf{V}_{\mathcal{L}}^{\ell}(R) 
&:=
\inf_{a} \vE_{X}\left[\ell(X, a)\right] 
-  \displaystyle \inf_{\substack{p_{A\mid X}\colon \\ 
{\mathcal{L}(X\to A)} \leq R}} \vE_{X, A}\left[\ell(X, A)\right].
\label{eq:arimoto_VoI}
\end{align}
Then $\textsf{V}_{\mathcal{L}}^{\ell}(0) = 0$ and 
for $0\leq R\leq K(X)$ and for arbitrary alphabet $\mathcal{Y}$, 
\begin{align}
\textsf{V}_{\mathcal{L}}^{\ell}(R; \mathcal{Y}) 
\leq \textsf{V}_{\mathcal{L}}^{\ell}(R). 
\label{eq:main_result_for_information_leakage}
\end{align}
Moreover, let $t(A)$ be a \textit{sufficient statistic of $A$ for $X$} and $t(\mathcal{A})$ be a set of all values of the statistic.
Then the equality in the inequality \eqref{eq:main_result_for_information_leakage} holds 
when $\mathcal{Y}= t(\mathcal{A})$ and the optimal conditional distribution is given by 
\begin{align}
p^{*}_{Y\mid X}(y\mid x) := \sum_{a}p^{*}_{A\mid X}(a\mid x) \one{y = t(a)}, 
\end{align}
where $p^{*}_{A\mid X} = \arginf_{p_{A\mid X}\colon \mathcal{L}(X\to A)\leq R}\vE_{X, A}\left[\ell(X, A)\right]$. 
\end{theorem}

\begin{remark}
Stratonovich call $\textsf{V}_{I}^{\ell}(R)$ 
as \textit{Value of Shannon's Information} in \cite[Chapter 9.3]{belavkin2020theory}. 
Thus we call 
$\textsf{V}_{I_{\alpha}^{\text{A}}}^{\ell}(R)$ 
(resp. $\textsf{V}_{I_{\alpha}^{\text{S}}}^{\ell}(R), \textsf{V}_{I_{\alpha}^{\text{C}}}^{\ell}(R)$) 
and $\textsf{V}_{\mathcal{L}_{f}}^{\ell}(R)$ 
as \textit{Value of Arimoto's (resp. Sibson's, Csisz\'ar's) Information} and 
\textit{Value of $f$-leakage}. 
\end{remark}

\begin{remark}
In \cite{https://doi.org/10.48550/arxiv.2201.11449}, we gave an interpretation of the Theorem \ref{thm:main_result_for_information_leakage} in terms of optimal privacy mechanism $p_{Y\mid X}^{*}$ in the privacy-utility trade-off problem. 
To construct the optimal mechanism $p_{Y\mid X}^{*}$, we need to construct $p_{A\mid X}^{*} = \arginf_{p_{A\mid X}\colon \mathcal{L}(X\to A)\leq R}\vE_{X, A}\left[\ell(X, A)\right]$. 
The algorithm for computing $\textsf{V}_{\mathcal{L}}^{\ell}(R)$ in Section \ref{sec:computation_VoI} allows us to obtain this distribution.
\end{remark}

\begin{prop}[\text{\cite[Prop. 7]{https://doi.org/10.48550/arxiv.2201.11449}}] \label{prop:basic_property_gen_VoI}
\ 
\begin{enumerate}
\item $\textsf{V}_{\mathcal{L}}^{\ell}(R)$ is increasing in $R$.
\item $\textsf{V}_{\mathcal{L}}^{\ell}(R)$ is concave (resp. quasi-concave) 
if $\mathcal{L}(X\to A)$ is convex (resp. quasi-convex) in $p_{A\mid X}$. 
\end{enumerate}
\end{prop}

\begin{cor}[\text{\cite[Cor. 1]{https://doi.org/10.48550/arxiv.2201.11449}}] From the property $2)$ above, the following hold: 
\begin{itemize}
\item $\textsf{V}^{\ell}_{I}(R)$ is concave (see Figure \ref{fig:VoI_alpha})
since $I(X; A)$ is convex in $p_{A\mid X}$ for fixed $p_{X}$ (see, e.g., \cite[Thm. 2.7.4]{Cover:2006:EIT:1146355})
\item $\textsf{V}^{\ell}_{\mathcal{L}_{f}}(R)$ is concave 
since $\mathcal{L}_{f}(X\to A)$ is convex in $p_{A\mid X}$\footnote{From the convexity of $f$-divergence \cite[Lem. 4.1]{Csiszar:2004:ITS:1166379.1166380}, 
\cite[Thm. 7.3]{polyanskiy_lecnotes_fdiv}, 
one can derive the convexity of $\mathcal{L}_{f}(X\to A)$ in $p_{A\mid X}$. } 
for fixed $p_{X}$ 
\item For $\alpha>0$, $\textsf{V}^{\ell}_{I_{\alpha}^{\text{A}}}(R)$ is quasi-concave 
since $I_{\alpha}^{\text{A}}(X; A)$ is quasi-convex in $p_{A\mid X}$ for fixed $p_{X}$  (see \cite[Footnote 3]{8804205}). 
For $0< \alpha \leq 1$, $\textsf{V}^{\ell}_{I_{\alpha}^{\text{A}}}(X; A)$ is concave since $I_{\alpha}^{\text{A}}(X; A)$ is convex in $p_{A\mid X}$ for fixed $p_{X}$ (see Proposition \ref{prop:convexity_arimoto})
\item For $\alpha > 0$, $\textsf{V}^{\ell}_{I_{\alpha}^{\text{S}}}(R)$ is quasi-concave 
since $I_{\alpha}^{\text{S}}(X; A)$ is quasi-convex in $p_{A\mid X}$ for fixed $p_{X}$. 
For $0 < \alpha\leq 1$, $\textsf{V}^{\ell}_{I_{\alpha}^{\text{S}}}(R)$ is concave 
since $I_{\alpha}^{\text{S}}(X; A)$ is convex in $p_{A\mid X}$ for fixed $p_{X}$ (see \cite[Thm. 10]{7282554})
\item For $0 < \alpha\leq 1$, $\textsf{V}^{\ell}_{I_{\alpha}^{\text{C}}}(R)$ is concave 
since $I_{\alpha}^{\text{C}}(X; A)$ is convex in $p_{A\mid X}$ for fixed $p_{X}$ (see \cite[Thm. 9 (c)]{e23020199})

\end{itemize}
\end{cor}

\begin{figure}[htbp]
\centering
\includegraphics[width=1.8in, clip]{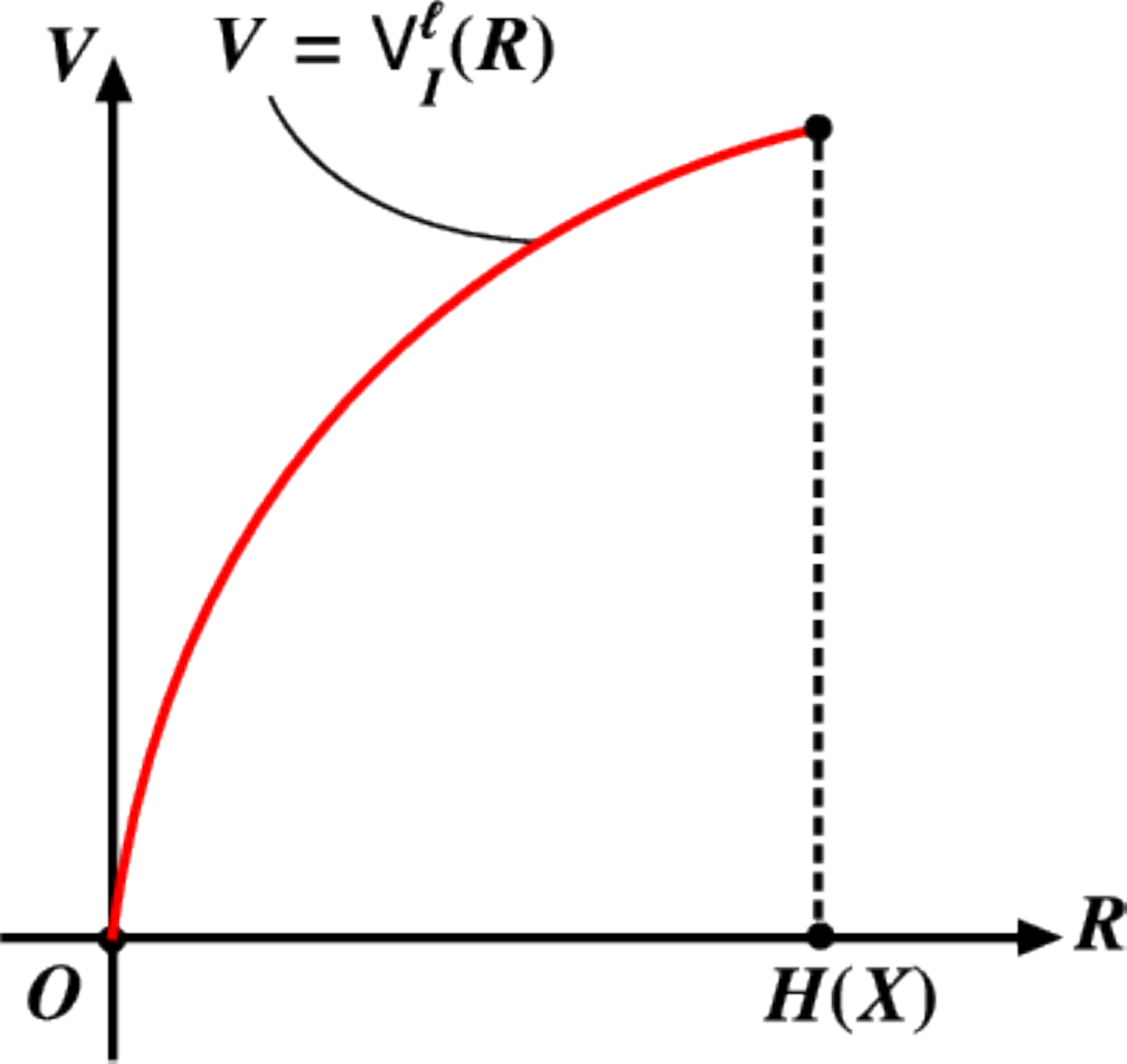}
\caption{Value of Shannon's information}
\label{fig:VoI_alpha}
\end{figure}

\begin{prop} \label{prop:equality_constant}
Assume that $\mathcal{L}(A\to X) = \mathcal{L}(p_{X}, p_{A\mid X})$ is convex in $p_{A\mid X}$ for fixed $p_{X}$. 
Then the inequality constraint in \eqref{eq:arimoto_VoI} for $\textsf{V}_{\mathcal{L}}^{\ell}(R)$ 
can be replaced by an equality constraint, i.e., the following holds: 
\begin{align}
\textsf{V}_{\mathcal{L}}^{\ell}(R) 
&= \inf_{a} \vE_{X}\left[\ell(X, a)\right] 
-  \displaystyle \inf_{\substack{p_{A\mid X}\colon \\ 
{\mathcal{L}(X\to A)} = R}} \vE_{X, A}\left[\ell(X, A)\right].
\label{eq:VoI_equality_constraint}
\end{align}
\end{prop}
\begin{proof}
See Appendix \ref{proof:equality_constant}. 
\end{proof}

\subsection{Alternating Optimization} \label{ssec:alt_alg}

Let $B_{i}$ be a convex subset of $\br^{n_{i}}$ for $i=1, 2$.
Let $f\colon B_{1}\times B_{2} \to \br$ be a continuous function defined on $B_{1}\times B_{2}$ that 
are bounded from below and has continuous partial derivatives $\nabla f = \left( {\partial f}/{\partial u_{1}}, {\partial f}/{\partial u_{2}} \right)$ on $B_{1}\times B_{2}$. 
Then consider the double infimum 
\begin{align}
f^{*} := \inf_{u_{1} \in B_{1}} \inf_{u_{2}\in B_{2}} f(u_{1}, u_{2}).
\end{align}

Assume that for all $u_{2}\in B_{2}$ there exists a unique $c_{1}(u_{2})\in B_{1}$ such that 
\begin{align}
f(c_{1}(u_{2}), u_{2}) = \min_{u_{1}^{\prime}} f(u_{1}^{\prime}, u_{2}). \label{eq:double_infimum}
\end{align}
Similarly, assume that for all $u_{1}\in B_{1}$ there exists a unique $c_{2}(u_{1})\in B_{2}$ such that 
\begin{align}
f(u_{1}, c_{2}(u_{1})) = \min_{u_{2}^{\prime}} f(u_{1}, u^{\prime}_{2}).
\end{align}

Let $u_{1}^{(0)}$ be an arbitrarily chosen vector in $B_{1}$. 
Then define a sequence $\{(u_{1}^{(k)}, u_{2}^{(k)})\}_{k=0}^{\infty}$ as follows:
\begin{align}
u_{1}^{(k)} &:= c_{1}(u_{2}^{(k-1)}), \\
u_{2}^{(k)} &:= c_{2}(u_{1}^{(k)}).
\end{align}

\begin{prop}[\text{\cite[Thm. 10.5]{10.5555/1199866}\footnotemark }] \label{prop:sufficient_condition}
If $f$ is convex on $B_{1}\times B_{2}$, 
then $f(u_{1}^{(k)}, u_{2}^{(k)})\to f^{*}$ as $k\to \infty$. 
\end{prop}
\footnotetext{Note that, instead of the double infimum problem \eqref{eq:double_infimum}, a double supremum problem is considered in \cite{10.5555/1199866}.}

\section{Computation of the Value of Information}\label{sec:computation_VoI}
Note that $\textsf{V}_{\mathcal{L}}^{\ell}(R)$ can be represented as 
\begin{align}
\textsf{V}_{\mathcal{L}}^{\ell}(R) = U_{\mathcal{L}}^{\ell}(0) - U_{\mathcal{L}}^{\ell}(R),
\end{align} 
where $U_{\mathcal{L}}^{\ell}(R):=\inf_{p_{A\mid X}\colon \mathcal{L}(X\to A)\leq R}\vE_{X, A}\left[\ell(X, A)\right]$. 
Thus the computation of $\textsf{V}_{\mathcal{L}}^{\ell}(R)$ results in the computation of $U_{\mathcal{L}}^{\ell}(R)$.
In this section, we provide an alternating optimization algorithm framework for computing 
$U_{\mathcal{L}}^{\ell}(R)$ and derive a convergence condition to a globally optimal solution. 

\begin{theorem} \label{thm:fundamental}
Assume that there exists a non-negative functional $G(p_{A\mid X}, q_{A})\geq 0$ such that 
\begin{align}
\mathcal{L}(p_{X}, p_{A\mid X}) 
&= \min_{q_{A}>0} G(p_{A\mid X}, q_{A}), 
\end{align}
where the minimum is over all distributions $q_{A}$ that satisfies $q_{A}(a)>0$ for all $a\in \mathcal{A}$. 
For $\beta\geq 0$, define a function $F_{\beta}(p_{A\mid X}, q_{A})$ 
and distributions $(p_{A\mid X}^{*}, q_{A}^{*})$ as follows:
\begin{align}
F_{\beta}(p_{A\mid X}, q_{A}) 
&:= \vE_{X, A}\left[\ell(X, A)\right] + \beta G(p_{A\mid X}, q_{A}), \\ 
F_{\beta}(p_{A\mid X}^{*}, q_{A}^{*}) &= \inf_{p_{A\mid X}\in B_{1}} \inf_{q_{A}\in B_{2}} F_{\beta}(p_{A\mid X}, q_{A}),  \label{eq:double_infimum_F}
\end{align}
where 
\begin{align}
B_{1} &= \left\{\, p_{A\mid X} \relmiddle{|} \forall a, x: p_{A\mid X}(a\mid x)>0, \sum_{a}p_{A\mid X}(a\mid x)=1 \right\}, \\ 
B_{2} &= \left\{\, q_{A} \relmiddle{|} \forall a>0: q_{A}(a)>0, \sum_{a}q_{A}(a)=1 \right\}.
\end{align}
Then, the following holds: 
\begin{align}
{U}_{\mathcal{L}}^{\ell} (R_{\beta}) + \beta R_{\beta} 
&= F_{\beta}(p_{A\mid X}^{*}, q_{A}^{*}), 
\end{align}
where $R_{\beta} := \mathcal{L}(p_{X}, p_{A\mid X}^{*})$ and 
\begin{align}
U_{\mathcal{L}}^{\ell}(R_{\beta}) := \inf_{\substack{p_{A\mid X}\colon \\ 
{\mathcal{L}(X\to A)} = R_{\beta}}} \vE_{X, A}\left[\ell(X, A)\right]. 
\end{align}
\end{theorem}
\begin{proof}
See Appendix \ref{proof:thm:fundamental}.
\end{proof}

\begin{cor} \label{cor:convergence} 
Assume that for all $p_{A\mid X}\in B_{1}$ there exists a unique $c_{2}(p_{A\mid X})\in B_{2}$ 
such that $F_{\beta}(p_{A\mid X}, c_{2}(p_{A\mid X})) = \min_{q_{A}\in B_{2}} F(p_{A\mid X}, q_{A})$. 
Similarly, assume that for all $q_{A}\in B_{2}$ there exists a unique $c_{1}(q_{A})\in B_{1}$ 
such that $F_{\beta}(c_{1}(q_{A}), q_{A}) = \min_{p_{A\mid X}\in B_{1}} F(p_{A\mid X}, q_{A})$. 
Let $p_{A\mid X}^{(0)}\in B_{1}$ be an arbitrary probability distribution on $A$ and 
define sequences $\{(q_{A}^{(k)}, p_{A\mid X}^{(k)})\}_{k=0}^{\infty}, \left\{ F^{(k)}\right\}_{k=0}^{\infty}$ as follows:
\begin{align}
q_{A}^{(k)} &:= \argmin_{q_{A}\in B_{2}} F_{\beta}(p_{A\mid X}^{(k-1)}, q_{A}) \\ 
p_{A\mid X}^{(k)} &:= \argmin_{p_{A\mid X}\in B_{1}} F_{\beta}(p_{A\mid X}, q_{A}^{(k)}), \\ 
F^{(k)} &:= F_{\beta}(p_{A\mid X}^{(k)}, q_{A}^{(k)}).
\end{align}
If $G(p_{A\mid X}, q_{A})$ is jointly convex on $B_{1}\times B_{2}$, then 
\begin{align}
F^{(k)} \to U_{\mathcal{L}}^{\ell}(R_{\beta}) + \beta R_{\beta},  \qquad \text{as $k\to \infty$}. 
\end{align}
\end{cor}
\begin{proof}
First, $F_{\beta}(p_{A\mid X}, q_{A})$ is bounded from below since $\ell(x, a)\geq 0$, $\beta\geq 0$ and $G(p_{A\mid X}, q_{A})\geq 0$. 
Moreover, since $\vE_{X, A}\left[\ell(X, A)\right]$ is linear (thus both convex and concave) on $B_{1}\times B_{2}$, the joint convexity of $F_{\beta}(p_{A\mid X}, q_{A})$ is equivalent to the joint convexity of $G(p_{A\mid X}, q_{A})$. 
Therefore, the proof is complete by applying Proposition \ref{prop:sufficient_condition} 
to $F_{\beta}(p_{A\mid X}, q_{A})$.
\end{proof}

%
From Theorem \ref{thm:fundamental} and Corollary \ref{cor:convergence}, the following Arimoto-Blahut-like alternating optimization algorithm is derived.

\begin{remark}
$U_{\mathcal{L}}^{\ell}(R)$ corresponds to the distortion-rate function $D(R):=\inf_{p_{Y\mid X} \colon I(X; Y)\leq R}\vE_{X, Y}\left[d(X, Y)\right]$, where $d(x, y)$  is a distortion funciton. 
Therefore, by replacing $\ell(x, a)$ with $d(x, y)$, 
the results above also hold for the generalized distortion-rate function defined as 
$D_{\mathcal{L}}(R):=\inf_{p_{Y\mid X} \colon \mathcal{L}(X\to Y)\leq R}\vE_{X, Y}\left[d(X, Y)\right]$.
\end{remark}

\begin{algorithm}[h]
	\caption{Arimoto--Blahut-like algorithm}
	\label{alg:ab_like}
	\begin{algorithmic}[1]
		\Require 
			\Statex $\epsilon>0, \beta\geq 0$
			\Statex $p_{A\mid X}^{(0)}\in B_{1}$ 
		\Ensure
			\Statex $U_{\mathcal{L}}^{\ell}(R_{\beta})$ 
		\Initialize{
			$F^{(-1)} \gets 0$ \\ 
			$q_{A}^{(0)} \gets \argmin_{q_{A}>0} F_{\beta}(p_{A\mid X}^{(0)}, q_{A})$ \\ 
			$F^{(0)}\gets F_{\beta}(p_{A\mid X}^{(0)}, q_{A}^{(0)})$ \\ 
			$k\gets 0$ \\
      }
		\While{$\abs{F^{(k)} - F^{(k-1)}} > \epsilon$}
			\State $p_{A\mid X}^{(k)} \gets \argmin_{p_{A\mid X}>0} F_{\beta}(p_{A\mid X}, q_{A}^{(k-1)})$ 
			\State $q_{A}^{(k)}\gets \argmin_{q_{A}>0} F_{\beta}(p_{A\mid X}^{(k)}, q_{A})$
			\State $F^{(k)} \gets F_{\beta}(p_{A\mid X}^{(k)}, q_{A}^{(k)})$
			\State $k\gets k+1$
		\EndWhile
		\State \textbf{return} $F^{(k)} - \beta\mathcal{L}(p_{X}, p_{A\mid X}^{(k)})$
	\end{algorithmic}
\end{algorithm}

\section{Applications}\label{sec:application}
In this section, we discuss alternating optimization algorithms for computing VoI under the constraint of a specific information leakage measure that includes Shannon's MI, 
$f$-leakage, Arimoto's MI, Sibson's MI and Csisz{\'a}r's MI.
\subsection{Computation of the Value of Shannon's Information}

\begin{prop}[\text{\cite[Lem. 10.8.1]{Cover:2006:EIT:1146355}}] \label{prop:min_relative_entropy}
\begin{align}
I(X; A)  
&= \min_{q_{A}>0} D(p_{X}p_{A\mid X} || p_{X}q_{A}), 
\end{align}
where the minimum is achieved at 
\begin{align}
q_{A}^{*}(a) &:= \sum_{x}p_{X}(x)p_{A\mid X}(a\mid x).
\end{align}
\end{prop}

\begin{prop} 
Let 
\begin{align}
&F_{\beta}(p_{A\mid X}, q_{A}) := 
\vE_{X,A}\left[\ell(X, A)\right] + \beta D(p_{X}p_{A\mid X} || p_{X}q_{A}). \label{eq:F_beta_Shannon}
\end{align}
Then 
\begin{enumerate}
\item For fixed $p_{A\mid X}$, $F_{\beta}(p_{A\mid X}, q_{A})$ is minimized by 
\begin{align}
q_{A}^{*}(a) &= \sum_{x}p_{X}(x)p_{A\mid X}(a\mid x).
\end{align}
\item For fixed $q_{A}$, $F_{\beta}(p_{A\mid X}, q_{A})$ is minimized by 
\begin{align}
p_{A\mid X}^{*}(a\mid x) 
&= \frac{q_{A}(a)e^{\frac{1}{\beta}\cdot \ell(x,a)}}{\sum_{a^{\prime}}q_{A}(a^{\prime})e^{\frac{1}{\beta}\cdot \ell(x,a^{\prime})}}. 
\end{align}
\end{enumerate}
\end{prop}

Since $G(p_{A\mid X}, q_{A}) := D(p_{A\mid X} || q_{A})$ is jointly convex 
on $(p_{A\mid X}, q_{A})$ (see, e.g., \cite[Thm. 2.7.2]{Cover:2006:EIT:1146355}), 
the following holds from Corollary \ref{cor:convergence}.
\begin{cor} 
The alternating optimization algorithm corresponding to this problem converges to $U_{I}^{\ell}(R_{\beta})$. 
\end{cor}

\subsection{Computation of Value of $f$-leakage} \label{ssec:f_leakage}

The minimizer of the $f$-leakage $\mathcal{L}_{f}(X\to A)$ depends on the function $f$. 
Since 
$G_{f}(p_{A\mid X}, q_{A}) := D_{f}(p_{X}p_{A\mid X} || p_{X} q_{A})$ 
is jointly convex on $(p_{A\mid X}, q_{A})$ (see \cite[Lem. 4.1]{Csiszar:2004:ITS:1166379.1166380}), the following propositions follow from the KKT condition.
%

\begin{prop}
Let $f\colon [0, \infty)\to \br$ be a differentiable convex function such that $f(1) = 0$ and strictly convex at $t=1$.
Any minimizers $q_{A}^{*}$ of the $f$-leakage $\mathcal{L}_{f}(X\to A) = \min_{q_{A}>0} D_{f}(p_{X}p_{A\mid X} || p_{A}q_{A})$ satisfy the following equation for some $\lambda$: 
\begin{align}
&\sum_{x}p_{X}(x) \Biggl\{f \left( \frac{p_{A\mid X}(a\mid x)}{q_{A}^{*}(a)} \right) \notag  \\  
&\qquad  - \left( \frac{p_{A\mid X}(a\mid x)}{q_{A}^{*}(a)} \right)\cdot f^{\prime}\left( \frac{p_{A\mid }(a\mid x)}{q_{A}^{*}(a)} \right) \Biggr\} + \lambda = 0. \label{eq:f_eqn}
\end{align}
\end{prop}

\begin{eg}
By solving this equation \eqref{eq:f_eqn} for $q_{A}^{*}$, the minimizers for each $f$ are obtained as shown in the third column of Table \ref{tab:f-leakage}.
\end{eg}

\begin{prop} 
Let $f\colon [0, \infty)\to \br$ be a differentiable convex function such that $f(1) = 0$ and strictly convex at $t=1$. Define 
\begin{align}
F^{\text{$f$}}_{\beta}(p_{A\mid X}, q_{A}) 
&:= \vE_{X,A}\left[\ell(X, A)\right] +\beta D_{f}(p_{X}p_{A\mid X} || p_{X} q_{A}). \label{eq:F_beta_f_leakage}
\end{align}
Then 
\begin{enumerate}
\item For fixed $p_{A\mid X}$, any minimizers $q_{A}^{*}$ of $\min_{q_{A}>0} F^{\text{$f$}}_{\beta}(p_{A\mid X}, q_{A})$ satisfy \eqref{eq:f_eqn} for some $\lambda$.

\item For fixed $q_{A}$, any minimizers $p_{A\mid X}^{*}$ of $\min_{p_{A\mid X}>0} F^{\text{$f$}}_{\beta}(p_{A\mid X}, q_{A})$ satisfy the following equation for all $x\in \mathcal{X}$ and for some $\lambda_{x}$: 
\begin{align}
\ell(x, a) + \beta f^{\prime} \left( \frac{p_{A\mid X}^{*}(a\mid x)}{q_{A}(a)} \right)  + \lambda_{x} = 0. \label{eq:f_p_A_X}
\end{align}
\end{enumerate}
\end{prop}

\begin{remark}
Unfortunately, the equation \eqref{eq:f_p_A_X} does not necessarily have an analytic solution or unique solution. 
For example, when $f(t) = 2(1-\sqrt{t})$, by solving \eqref{eq:f_p_A_X} we have
\begin{align}
p_{A\mid X}^{*}(a\mid x) &= \frac{\beta^{2}q_{A}(a)}{(\ell(x, a) + \lambda_{x})^{2}}, \\ 
1 &= \sum_{a} \frac{\beta^{2}q_{A}(a)}{(\ell(x, a) + \lambda_{x})^{2}},
\end{align}
which does not have an analytic solution or unique solution in general. 
\end{remark}

%

\subsection{Computation of Value of Arimoto's Information}

\begin{prop}[\text{\cite[Thm. 4]{6898022}}]
\begin{align}
I_{\alpha}^{\text{A}}(X; A) &= \min_{q_{A}>0} \left\{ D_{\alpha}(p_{X}p_{A\mid X} || u_{X}q_{A}) - D_{\alpha}(p_{X} || u_{X}) \right\}, \label{eq:arimoto_min}
\end{align}
where $u_{X}(x) := 1/\abs{\mathcal{X}}$ for all $x\in \mathcal{X}$. 
Moreover, \eqref{eq:arimoto_min} is minimized by 
\begin{align}
q_{A}^{*}(a) &= \frac{\left\{\sum_{x}p_{X}(x)^{\alpha}p_{A\mid X}(a\mid x)^{\alpha}\right\}^{1/\alpha}}{\sum_{a} \left\{\sum_{x}p_{X}(x)^{\alpha}p_{A\mid X}(a\mid x)^{\alpha}\right\}^{1/\alpha}}. 
\end{align}
\end{prop}

Making use of this result and the fact that $D_{\alpha}(p || q)$ is jointly convex on $(p, q)$ for $0 \leq \alpha \leq 1$ (see \cite[Thm. 11]{6832827}), we can prove convexity of the 
Arimoto's mutual information of order $\alpha$ on $p_{A\mid X}$ for fixed $p_{X}$, for $0 < \alpha \leq 1$. 

\begin{prop} \label{prop:convexity_arimoto}
For $0 < \alpha \leq 1$, $I_{\alpha}^{\text{A}}(X; A) = I_{\alpha}^{\text{A}}(p_{X}, p_{A\mid X})$ is convex in $p_{A\mid X}$ for fixed $p_{X}$.
\end{prop}
\begin{proof}
See Appendix \ref{proof:convexity_arimoto}. 
\end{proof}
Joint convexity of $G^{\text{A}}_{\alpha}(p_{A\mid X}, q_{A}) := D_{\alpha}(p_{X}p_{A\mid X} || u_{X}q_{A}) - D_{\alpha}(p_{X} || u_{X})$ on $(p_{A\mid X}, q_{A})$ 
follows immediately from \cite[Thm. 11]{6832827}. Then, 
as in section \ref{ssec:f_leakage}, the following proposition follows from the KKT condition.
\begin{prop} 
Let 
\begin{align}
F^{\text{A}, \alpha}_{\beta}(p_{A\mid X}, q_{A}) &:=
\vE_{X,A}\left[\ell(X, A)\right] \notag \\ 
& + \beta \left\{ D_{\alpha}(p_{X}p_{A\mid X} || u_{X}q_{A}) - D_{\alpha}(p_{X} || u_{X}) \right\}. \label{eq:F_beta_Arimoto}
\end{align}
Then 
\begin{enumerate}
\item For fixed $p_{A\mid X}$, $F^{\text{A}, \alpha}_{\beta}(p_{A\mid X}, q_{A})$ is minimized by 
\begin{align}
q_{A}^{*}(a) &= \frac{\left\{\sum_{x}p_{X}(x)^{\alpha}p_{A\mid X}(a\mid x)^{\alpha}\right\}^{1/\alpha}}{\sum_{a} \left\{\sum_{x}p_{X}(x)^{\alpha}p_{A\mid X}(a\mid x)^{\alpha}\right\}^{1/\alpha}}. 
\end{align}
\item For fixed $q_{A}$, 
any minimizers $p_{A\mid X}^{*}$ of $\min_{p_{A\mid X}>0} F^{\text{A}, \alpha}_{\beta}(p_{A\mid X}, q_{A})$ satisfy the following equation for all $x\in \mathcal{X}$ and for some $\lambda_{x}$: 
\begin{align}
&\ell(x, a) +\frac{\alpha \beta}{\alpha-1}  \notag \\ 
&\times\frac{p_{X}(x)^{\alpha-1}p_{A\mid X}^{*}(a\mid x)^{\alpha-1}q_{A}(a)^{1-\alpha}}{\sum_{x, a} p_{X}(x)^{\alpha}p_{A\mid X}^{*}(a\mid x)^{\alpha}q_{A}(a)^{1-\alpha}} + \lambda_{x} = 0.
\label{eq:Arimoto_p_A_X}
\end{align}
\end{enumerate}
\end{prop}

\begin{remark}
Unfortunately, the equation\eqref{eq:Arimoto_p_A_X} does not necessarily have an analytic solution or a unique solution.
\end{remark}

%

\subsection{Computation of Value of Sibson's Information}
\begin{prop}[\text{\cite[Eq. (11)]{370121}}]
\begin{align}
I_{\alpha}^{\text{S}}(X; A) &:= \min_{q_{A}>0} D_{\alpha}(p_{X}p_{A\mid X} || p_{X}q_{A}) 
\end{align}
is minimized by 
\begin{align}
q^{*}_{A} (a) &= \frac{\left\{\sum_{x}p_{X}(x)p_{A\mid X}(a\mid x)^{\alpha} \right\}^{1/\alpha}}{\sum_{a}\left\{\sum_{x}p_{X}(x)p_{A\mid X}(a\mid x)^{\alpha} \right\}^{1/\alpha}}. \label{eq:minimizer_q_A_Sibson}
\end{align}
\end{prop}

\begin{remark} 
The minimizer $q_{A}^{*}$ in \eqref{eq:minimizer_q_A_Sibson} is as same as the minimizer 
of $f$-leakage $\min_{q_{A}>0}D_{f}(p_{X}p_{X\mid A} || p_{X}q_{A})$ when $f(t)=f_{\alpha}(t):=(t^{\alpha}-1)/(\alpha-1)$ in Table \ref{tab:f-leakage}.
This can also be confirmed via the fact that $D_{\alpha}(p || q) = (\alpha-1)^{-1}\log (1+ (\alpha-1)D_{f_{\alpha}}(p || q))$ \cite[Eq. (80)]{7552457}. 
\end{remark}

Since $G^{\text{S}}_{\alpha}(p_{A\mid X}, q_{A}) := D_{\alpha}(p_{X}p_{A\mid X} || p_{X}q_{A})$ is jointly convex on $(p_{A\mid X}, q_{A})$ (see \cite[Thm. 11]{6832827}), 
the following proposition follows from the KKT condition. 
\begin{prop} 
Let 
\begin{align}
F^{\text{S}, \alpha}_{\beta}(p_{A\mid X}, q_{A}) &:= 
&\vE_{X,A}\left[\ell(X, A)\right] + \beta D_{\alpha}(p_{X}p_{A\mid X} || p_{X}q_{A}). \label{eq:F_beta_Sibson}
\end{align}
Then 
\begin{enumerate}
\item For fixed $p_{A\mid X}$, $F^{\text{S}, \alpha}_{\beta}(p_{A\mid X}, q_{A})$ is minimized by 
\begin{align}
q_{A}^{*}(a) &= \frac{\left\{\sum_{x}p_{X}(x)p_{A\mid X}(a\mid x)^{\alpha} \right\}^{1/\alpha}}{\sum_{a}\left\{\sum_{x}p_{X}(x)p_{A\mid X}(a\mid x)^{\alpha} \right\}^{1/\alpha}}.
\end{align}
\item For fixed $q_{A}$, 
any minimizers $p_{A\mid X}^{*}$ of $\min_{p_{A\mid X}>0} F^{\text{S}, \alpha}_{\beta}(p_{A\mid X}, q_{A})$ satisfy the following equation for all $x\in \mathcal{X}$ and for some $\lambda_{x}$: 
\begin{align}
&\ell(x, a)  + \frac{\alpha \beta}{\alpha-1} \notag \\ 
&\times \frac{p_{A\mid X}^{*}(a\mid x)^{\alpha-1}q_{A}(a)^{1-\alpha}}{\sum_{x,a}p_{A\mid X}^{*}(a\mid x)^{\alpha-1}q_{A}(a)^{1-\alpha}} + \lambda_{x} = 0. \label{eq:Sibson_p_A_X}
\end{align}
\end{enumerate}
\end{prop}

\begin{remark}
Unfortunately, the equation \eqref{eq:Sibson_p_A_X} does not necessarily have an analytic solution or a unique solution.
\end{remark}


\subsection{Computation of Value of Csisz\'ar's Information}
Since joint convexity of $G^{\text{C}}_{\alpha}(p_{A\mid X}, q_{A}) := \vE_{X}\left[D_{\alpha}(p_{A\mid X}(\cdot\mid X) || q_{A})\right]$ on $(p_{A\mid X}, q_{A})$ follows immediately from \cite[Thm. 11]{6832827}, 
the following propositions follow from the KKT condition.

\begin{prop}
Any minimizers $q_{A}^{*}$ of the Csisz\'ar's MI $I_{\alpha}^{\text{C}}(X; A) := \min_{q_{A}} \vE_{X}\left[D_{\alpha}(p_{A\mid X}(\cdot\mid X) || q_{A})\right]$ satisfy the following equation for some $\lambda$: 
\begin{align}
& \sum_{x}p_{X}(x) \cdot \frac{p_{A\mid X}(a\mid x)q_{A}^{*}(a)^{-\alpha}}{\sum_{a}p_{A\mid X}(a\mid x)q_{A}^{*}(a)^{1-\alpha}} - \lambda = 0. \label{eq:C_eqn}
\end{align}
\end{prop}

\begin{remark}
The equation \eqref{eq:C_eqn}
does not have analytical solution in general. Instead, we can use an iterative algorithm proposed by Karakos \textit{et al.} \cite{4595361} for a numerical solution.
\end{remark}

\begin{prop} 
Let 
\begin{align}
&F^{\text{C}, \alpha}_{\beta}(p_{A\mid X}, q_{A}) \notag \\
&:= \vE_{X,A}\left[\ell(X, A)\right] + \beta \vE_{X}\left[D_{\alpha}(p_{A\mid X}(\cdot\mid X) || q_{A})\right]. \label{eq:F_beta_Csiszar}
\end{align}
Then 
\begin{enumerate}
\item For fixed $p_{A\mid X}$, 
any minimizers $q_{A}^{*}$ of $\min_{p_{A\mid X}} F^{\text{C}}_{\beta}(p_{A\mid X}, q_{A})$ satisfy the equation \eqref{eq:C_eqn}. 
\item For fixed $q_{A}$, 
any minimizers $p_{A\mid X}^{*}$ of $\min_{p_{A\mid X}} F^{\text{C}}_{\beta}(p_{A\mid X}, q_{A})$ satisfy the following equation for all $x\in \mathcal{X}$ and for some $\lambda_{x}$: 
\begin{align}
&\ell(x, a) + \frac{\alpha \beta}{\alpha-1} \notag \\ 
&\times \frac{p_{A\mid X}^{*}(a\mid x)^{\alpha-1}q_{A}(a)^{1-\alpha}}{\sum_{a}p_{A\mid X}^{*}(a\mid x)^{\alpha-1}q_{A}(a)^{1-\alpha}} + \lambda_{x} = 0. \label{eq:Csiszar_p_A_X}
\end{align}
\end{enumerate}
\end{prop}

\begin{remark}
Unfortunately, the equation \eqref{eq:Csiszar_p_A_X} does not necessarily have an analytic solution or a unique solution.
\end{remark}

\section{Conclusion}\label{sec:conclusion}
In this study, we proposed an alternating optimization algorithm for computing the 
Stratonovich's value of information $\textsf{V}_{\mathcal{L}}^{\ell}(R)$ under a general information leakage constraint.
We also derived a convergence condition to globally optimal solution.
Future work includes constructing algorithms for solving numerically the equations for the KKT conditions in value of $f$-leakage, Arimoto information, Sibson information, Csisz{\'a}r information and conducting numerical experiments.

\appendices

\section{Proof of Proposition \ref{prop:equality_constant}}\label{proof:equality_constant}\begin{proof}
Based on \cite[Thm. 8.8]{youben:1978aa_en} and \cite[Cor. 9.19]{10.5555/1199866}, we will prove as follows.
First we will prove that $\textsf{V}_{\mathcal{L}}^{\ell}(R)$ is strictly increasing for $0\leq R\leq K(X)$. 
Define a function $U_{\mathcal{L}}^{\ell}(R)$ as
\begin{align}
U_{\mathcal{L}}^{\ell}(R) &:= \inf_{\substack{p_{A\mid X}\colon \\ 
{\mathcal{L}(X\to A)} \leq R}} \vE_{X, A}\left[\ell(X, A)\right]. \label{eq:U_L_l}
\end{align}
Then it suffices to show that $U_{\mathcal{L}}^{\ell}(R)$ is strictly decreasing.
To this end, we consider the following ``inverse'' function $R_{\mathcal{L}}^{\ell}(U)$ of $U_{\mathcal{L}}^{\ell}(R)$ defined as follows (see Figure \ref{fig:U_R}) and show that the function 
$R_{\mathcal{L}}^{\ell}(U)$ is strictly decreasing: 

\begin{align}
R_{\mathcal{L}}^{\ell}(U) &:=
\inf_{\substack{p_{A\mid X}\colon \\ 
{\vE_{X, A}\left[\ell(X, A)\right]} \leq U}} \mathcal{L}(X\to A).
\end{align}
Suppose that the function $R_{\mathcal{L}}^{\ell}(U)$ is constant on 
some interval $U^{\prime}\leq U\leq U^{\prime\prime}$. 
Let $\tilde{p}_{A\mid X}^{\circ}$ and $\tilde{p}_{A\mid X}^{\prime}$ be 
conditional distributions that achieve the infimum in $R_{\mathcal{L}}^{\ell}(U_{\text{max}})$ and  
$R_{\mathcal{L}}^{\ell}(U^{\prime})$, respectively, where $U_{\text{max}}:=\inf_{a}\vE_{X}\left[\ell(X, a)\right]$. 
Now, define $U_{\epsilon}:= (1- \epsilon) U^{\prime} + \epsilon U_{\text{max}}$ such that 
$U^{\prime} < U_{\epsilon} < U^{\prime \prime}$ for some $\epsilon>0$ and 
define a conditional distribution 
$\tilde{p}_{A\mid X}^{\epsilon}:= \epsilon \tilde{p}_{A\mid X}^{\circ} + (1-\epsilon)\tilde{p}_{A\mid X}^{\prime}$. 
Since the conditional distribution satisfies 
$\vE_{X, A}^{p_{X}\tilde{p}_{A\mid X}^{\epsilon}}\left[\ell(X, A)\right] \leq U_{\epsilon}$, 
then the following holds:
\begin{align}
R_{\mathcal{L}}^{\ell}(U_{\epsilon}) &:= 
\inf_{\substack{p_{A\mid X}\colon \\ 
{\vE_{X, A}\left[\ell(X, A)\right]} \leq U_{\epsilon}}} \mathcal{L}(X\to A) \\
&\leq \mathcal{L}(p_{X}, \tilde{p}_{A\mid X}^{\epsilon}) \\ 
&\overset{(a)}{\leq} \epsilon \mathcal{L}(p_{X}, \tilde{p}_{A\mid X}^{\circ}) + (1-\epsilon) \mathcal{L}(p_{X}, \tilde{p}_{A\mid X}^{\prime}) \\ 
&\overset{(b)}{=} (1-\epsilon)\mathcal{L}(p_{X}, \tilde{p}_{A\mid X}^{\prime}) \\
&< \mathcal{L}(p_{X}, \tilde{p}_{A\mid X}^{\prime}) = R_{\mathcal{L}}^{\ell}(U^{\prime}), 
\end{align}
where 
\begin{itemize}
\item $(a)$ follows from the assumption that $\mathcal{L}(p_{X}, p_{A\mid X})$ is convex in $p_{A\mid X}$, 
\item $(b)$ follows from $\mathcal{L}(p_{X}, \tilde{p}_{X\mid A}^{\circ}) = 0$.
\end{itemize}
This contradicts the assumption that 
$R_{\mathcal{L}}^{\ell}(U)$ is constant on $U^{\prime}\leq U\leq U^{\prime\prime}$.

Next, we will prove by contradiction that the inequality constraint in \eqref{eq:arimoto_VoI} 
can be replaced by an equality constraint. 
Suppose that $\textsf{V}_{\mathcal{L}}^{\ell}(R)$ is achieved by some $\tilde{p}_{A\mid X}$ such that $\mathcal{L}(p_{X}, \tilde{p}_{A\mid X}) =: R^{\prime} < R$. 
Then 
\begin{align}
\textsf{V}_{\mathcal{L}}^{\ell}(R^{\prime}) &\geq \inf_{a} \vE\left[\ell(X, a)\right] - \vE_{X, A}^{p_{X}\tilde{p}_{A\mid X}}\left[\ell(X, A)\right] \\ 
&= \textsf{V}_{\mathcal{L}}^{\ell}(R).
\end{align}
This contradicts the fact that $\textsf{V}_{\mathcal{L}}^{\ell}(R)$ is strictly increasing function for $0\leq R\leq K(X)$.
\begin{figure}[htbp]
\centering
\includegraphics[width=3.in, clip]{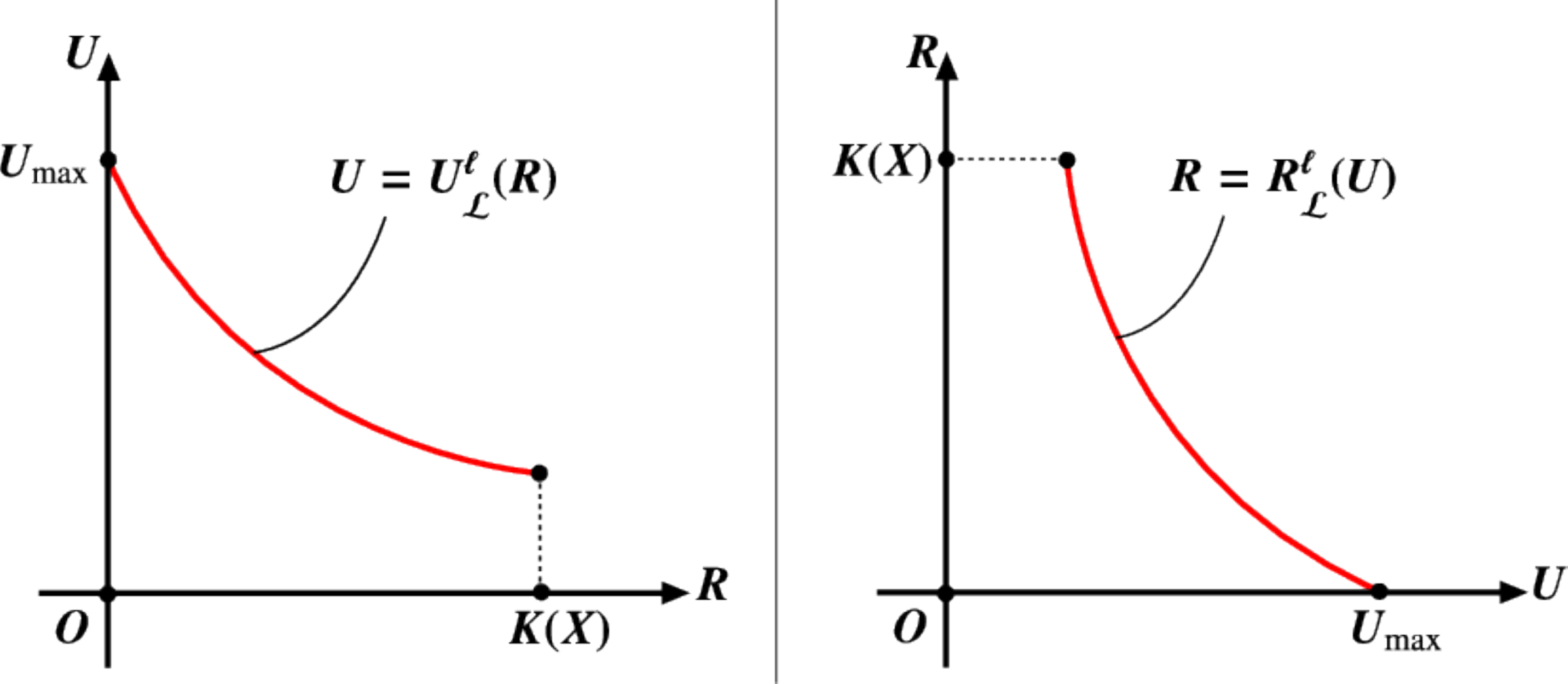}
\caption{$U_{\mathcal{L}}^{\ell}(R)$ (left) and $R_{\mathcal{L}}^{\ell}(U)$ (right)}
\label{fig:U_R}
\end{figure}
\end{proof}

\section{Proof of Theorem \ref{thm:fundamental}} \label{proof:thm:fundamental}
\begin{proof} 
Based on \cite[Thm. 8.8]{youben:1978aa_en} we will prove as follows.
Note that 
\begin{align}
&-\beta R_{\beta} + F_{\beta}(p_{A\mid X}^{*}, q_{A}^{*}) \notag \\
&= -\beta \mathcal{L}(p_{X}, p_{A\mid X}^{*}) + \vE_{X,A}^{p_{X}p_{A\mid X}^{*}}\left[\ell(X, A)\right] + \beta G(p_{A\mid X}^{*}, q_{A}^{*}) \\ 
&= \vE_{X,A}^{p_{X}p_{A\mid X}^{*}}\left[\ell(X, A)\right], 
\end{align}
where we used $\mathcal{L}(p_{X}, p_{A\mid X}^{*})=G(p_{A\mid X}^{*}, q_{A}^{*})$ in the second equality, since 
minimizing $F_{\beta}(p_{A\mid X}^{*}, q_{A})$ on $q_{A}$ is equivalent to minimizing $G(p_{A\mid X}^{*}, q_{A})$ on $q_{A}$. 
Then it suffices to show that $\vE_{X,A}^{p_{X}p_{A\mid X}^{*}}\left[\ell(X, A)\right] = U_{\mathcal{L}}^{\ell}(R_{\beta})$.
We will prove this by contradiction. 
Suppose that $\vE_{X,A}^{p_{X}p_{A\mid X}^{*}}\left[\ell(X, A)\right] > U_{\mathcal{L}}^{\ell}(R_{\beta})$.
Let $p_{A\mid X}^{\circ}$ be the distribution that achieves infimum in $U_{\mathcal{L}}^{\ell}(R_{\beta})$ 
and let $q_{A}^{\circ}$ be the distribution that achieves minimum in $\min_{q_{A}>0}G(p_{A\mid X}^{\circ}, q_{A}) = \mathcal{L}(p_{X}, p_{A\mid X}^{\circ})$. 
Since the distribution $p_{A\mid X}^{\circ}$ satisfies that 
$\vE_{X,A}^{p_{X}p_{A\mid X}^{\circ}}\left[\ell(X, A)\right] < \vE_{X,A}^{p_{X}p_{A\mid X}^{*}}\left[\ell(X, A)\right]$ 
and $\mathcal{L}(p_{X}, p_{A\mid X}^{\circ}) = R_{\beta}$ from Proposition \ref{prop:equality_constant}, 
the following holds:
\begin{align}
F_{\beta}(p_{A\mid X}^{\circ}, q_{A}^{\circ}) 
&= \vE_{X, A}^{p_{X}p_{X\mid A}^{\circ}}\left[\ell(X, A)\right] + \beta G(p_{A\mid X}^{\circ}, q_{A}^{\circ}) \\
&< \vE_{X,A}^{p_{X}p_{A\mid X}^{*}}\left[\ell(X, A)\right] + \beta R_{\beta} = F_{\beta}(p_{A\mid X}^{*}, q_{A}^{*}). 
\end{align}
This contradicts \eqref{eq:double_infimum_F}. 
\end{proof}

\section{Proof of Proposition \ref{prop:convexity_arimoto}}\label{proof:convexity_arimoto}
\begin{proof}\footnote{This proof is based on \cite[Thm. 10]{7282554}. }
For $\alpha=1$, the convexity of $I_{1}^{\text{A}}(X; A) = I(X; A)$ is well-known (see, e.g., \cite[Thm. 2.7.4]{Cover:2006:EIT:1146355}). 
For $0 < \alpha < 1$, let  $\lambda, \bar{\lambda}:=1-\lambda \in [0, 1]$ be arbitrary numbers 
and $p_{A\mid X}^{(1)}, p_{A\mid X}^{(2)}$ be arbitrary distributions. 
Define 
\begin{align}
q_{A}^{*, (i)} &:= \argmin_{q_{A}>0} \left\{ D_{\alpha}(p_{X}p_{A\mid X}^{(i)} || u_{X}q_{A}) - D_{\alpha}(p_{X} || u_{X}) \right\}
\end{align}
for $i=1, 2$. Then 
\begin{align}
&I_{\alpha}^{\text{A}}(p_{X}, \lambda p_{A\mid X}^{(1)} + \bar{\lambda}p_{A\mid X}^{(2)}) \notag \\ 
&= \min_{q_{A}} \Bigl\{ D_{\alpha} \left( p_{X}\left\{\lambda p_{A\mid X}^{(1)} + \bar{\lambda} p_{A\mid X}^{(2)} \right\} \relmiddle{|}\relmiddle{|} u_{X}q_{A}\right) \notag \\ 
&\qquad \qquad \qquad \qquad \qquad \qquad \qquad  - D_{\alpha}(p_{X} || u_{X} ) \Bigr\} \\ 
&\leq D_{\alpha} \left( p_{X}\left\{\lambda p_{A\mid X}^{(1)} + \bar{\lambda} p_{A\mid X}^{(2)} \right\} \relmiddle{|}\relmiddle{|} u_{X}\left\{\lambda q_{A}^{*, (1)} + \bar{\lambda} q_{A}^{*, (2)} \right\} \right) \notag \\ 
&\qquad \qquad \qquad \qquad \qquad \qquad \qquad \qquad \qquad - D_{\alpha}(p_{X} || u_{X}) \\ 
&\leq \lambda D_{\alpha} \left( p_{X}p_{A\mid X}^{(1)} \relmiddle{|}\relmiddle{|} u_{X}q_{A}^{*, (1)} \right) + \bar{\lambda} D_{\alpha} \left( p_{X}p_{A\mid X}^{(2)} \relmiddle{|}\relmiddle{|} u_{X}q_{A}^{*, (2)} \right) \notag \\ 
&\qquad \qquad \qquad \qquad \qquad \qquad \qquad \qquad \qquad- D_{\alpha}(p_{X} || u_{X}) \\ 
&= \lambda I_{\alpha}^{\text{A}}(p_{X}, p_{A\mid X}^{(1)}) + \bar{\lambda} I_{\alpha}^{\text{A}}(p_{X}, p_{A\mid X}^{(2)}), 
\end{align}
where second inequality follows from the fact that 
$D_{\alpha}(p || q)$ is jointly convex on $(p, q)$ for $0 \leq \alpha \leq 1$ (see \cite[Thm 11]{6832827}). 
\end{proof}

\clearpage


\end{document}